\documentclass[11pt]{article}

\usepackage[margin=1.1in]{geometry}
\usepackage{amsmath,amssymb,amsthm}
\usepackage{graphicx}
\usepackage{booktabs}
\usepackage{tabularx}
\usepackage[round]{natbib}
\usepackage{xcolor}
\usepackage[colorlinks=true,linkcolor=blue!50!black,citecolor=blue!50!black,urlcolor=blue!50!black]{hyperref}
\usepackage{enumitem}
\usepackage{microtype}

\newtheorem{proposition}{Proposition}
\newtheorem{lemma}{Lemma}
\newtheorem{corollary}{Corollary}
\theoremstyle{definition}
\newtheorem{definition}{Definition}
\newtheorem{assumption}{Assumption}
\theoremstyle{remark}
\newtheorem{remark}{Remark}

\newcolumntype{Y}{>{\raggedright\arraybackslash}X}

\title{\textbf{Strategic Technical Debt: A Real Options Approach to
Early-Stage Software Experimentation}}

\author{Rashid Azarang\\
\small Independent Researcher $\cdot$ Mentu, San Pedro Garza Garc\'ia, Nuevo Le\'on, Mexico\\
\small \texttt{rashid@mentu.ai} $\cdot$ ORCID 0009-0008-5528-4246 \and
Mohammad Reza Azarang Esfandiari, Ph.D.\\
\small Professor Emeritus, Departamento de Ingenier\'ia Industrial y de Sistemas,\\
\small Tecnol\'ogico de Monterrey, Campus Monterrey, Monterrey, Mexico\\
\small \texttt{mazarang@itesm.mx} $\cdot$ ORCID 0009-0005-7413-7610}

\date{\small Draft v0.5 $\cdot$ August 16, 2026.\\
\small Working paper. Theory with a specified (not executed) empirical
program; contains no empirical results.}

\begin{document}
\maketitle

\begin{abstract}
\noindent Technical debt is treated, almost universally, as an engineering
pathology: a liability incurred through haste and repaid through
suffering. This paper argues that under the conditions that define
early-stage software work (high hypothesis uncertainty, cheap
experiments, and the freedom to abandon), deliberately incurred
technical debt is not a pathology but a rationally priced financial
instrument: a call option on the validated product, purchased at a
discount that is largest exactly when uncertainty is highest. We make
three contributions. First, a \emph{demarcation}: debt is
\emph{strategic}\footnote{``Strategic debt'' also names an unrelated 2026 corporate-strategy construct (deferred strategy fundamentals; Zenodo record 21582207) that carries the debt metaphor out of software; the two frameworks share a phrase and nothing else.} when its expected cost loads on the success branch of
the venture (repaid only if the hypothesis validates) and \emph{toxic}
when it imposes unconditional cost while held (security exposure, data
loss, corrupted experimental signal), a boundary we state formally and
that renders the popular ``prudent vs.\ reckless'' intuition testable.
Second, a \emph{sequential model}: a finite-horizon dynamic program over
belief and debt stock whose solution yields four results: a \emph{shadow price of debt} $\varphi<1$ equal to the risk-discounted
probability of repayment; a \emph{technical-debt overhang} (the belief
threshold for scaling rises with the debt stock, an exercise-threshold
comparative static in the tradition of \citealp{mcdonald_siegel_1986},
related to but mechanistically distinct from \citealp{myers_1977}); a
\emph{refactoring-pivot theorem} (absent carrying costs, optimal
repayment concentrates at the commitment boundary, predicting a
refactoring burst at product--market fit, a pattern practitioners
report but, to our knowledge, no repository study has measured,
registered here as a falsifiable prediction); and a \emph{volatility
result} (the debt build holds a call where the robust build holds the
underlying, so under risk-neutral valuation mean-preserving spreads in
outcome value favor debt). A
\emph{pivot-salvage correction} shows the folk rule ``maximum debt at
maximum uncertainty'' is wrong whenever failure redirects rather than
terminates the venture and the salvage differential clears the
discounted cost premium ($\Delta\Pi > (C_r - C_s)/\delta$). Third, a set of falsifiable propositions with a
two-test primary empirical program advanced for pre-registration and
execution: validation-event refactoring timing with a
funding-confound design, and the first repository-history measurement
of pivot salvage, with three further propositions specified as
successors. Quoted quantities from the model are illustrative
calibration, not estimates. \medskip

\noindent\textbf{Keywords:} technical debt; real options; software
experimentation; startups; lean methodology; debt overhang;
refactoring; pre-registration.
\end{abstract}

% =====================================================================
\section{Introduction}\label{sec:intro}

The technical-debt metaphor entered software engineering as a
\emph{defense} of expedience: \citet{cunningham_1992} coined it to
explain to non-engineers why shipping ``not-quite-right code'' first
could be the correct move, provided the debt was later repaid. Three
decades of subsequent research inverted the valence. The literature
measures debt's \emph{interest} (degraded velocity, defect risk,
mounting rework; \citealp{kruchten_2012,tom_2013,besker_2018}) and
manages its \emph{repayment}
\citep{guo_seaman_2011,ampatzoglou_2015,yli_huumo_2016}. Practitioner
discourse followed: debt is confessed, deplored, and scheduled for
elimination. The moral framing is so settled that the field's central
taxonomies (deliberate vs.\ inadvertent, prudent vs.\ reckless;
\citealp{fowler_2009,mcconnell_2007}) exist mainly to separate
forgivable debt from blameworthy debt.

This paper takes Cunningham's original position seriously and asks what
it takes to make it \emph{rigorous} rather than rhetorical. The setting
is early-stage software work: a venture (a startup, a new product line,
a research prototype) holds a hypothesis $H$ about what the market
wants, assigns it a validation probability $p$ that is genuinely low,
and can buy information about $H$ by building and shipping an
experiment. The build can be \emph{robust} (production-grade, at cost
$C_r$) or \emph{debt-laden} (expedient, at cost $C_s \ll C_r$),
leaving behind a repayment obligation $D$. The observation that anchors
everything that follows is elementary but underexploited: \emph{the
repayment obligation is contingent}. If the hypothesis is refuted, the
code is discarded and the debt is never repaid. The venture holds a
call option on the validated product with strike $D$, plus a default
option that a financial borrower does not have: walking away costs
nothing, because the creditor is the venture's own future, a future
that, on the refuted branch, never arrives.

Three literatures each hold a piece of this argument, and none holds
the whole. Real-options analysis entered software design with
\citet{sullivan_1999} and \citet{baldwin_clark_2000}, but prices
\emph{flexibility of design} (modularity as embedded options), not the
\emph{debt instrument} itself. The economics of technical debt has
expected-cost models (\citealp{schmid_2013} already includes the
probability that a debt item ever requires repayment) but treats
uncertainty as a parameter to average over, not as the source of an
exercise policy: there is no abandonment decision, no repayment timing,
no derived predictions about \emph{when} debt is repaid. And the
entrepreneurship-as-experimentation literature
\citep{kerr_2014,camuffo_2020,koning_2022,gans_2019} establishes that
early ventures rationally buy information through cheap staged
experiments, and that startups in fact accumulate debt deliberately
while doing so \citep{giardino_2016,klotins_2018}, but has no formal
account of the engineering instrument that finances the purchase. The
contribution of this paper is the coupling: we price the debt
instrument inside the experimentation problem and extract its testable
consequences.

\subsection{Results preview}

Section~\ref{sec:construct} defines the construct. The load-bearing
definition is a \emph{cost-loading} criterion: a debt item is
\textbf{strategic} when its expected cost is carried by the success
branch (repayment happens only if the venture survives to scale) and
\textbf{toxic} when it imposes cost unconditionally while held: security exposure, unprotected data, compliance breach, or corruption
of the experimental signal itself. This is a sharper boundary than
prudent/reckless because it is a property of the \emph{payoff
structure}, checkable per item, rather than of the decision-maker's
state of mind, and it yields the paper's necessary-condition
prediction: at no level of hypothesis uncertainty is a toxic shortcut
preferable to taking the same shortcut without the hazard.

Section~\ref{sec:twoperiod} develops the two-period model.
Proposition~\ref{prop:rule} gives the exercise-adjusted debt rule
\[
C_r - C_s \;>\; \delta\!\left[\,p\,D + (1-p)\,\Delta\Pi\,\right],
\]
which strictly tightens the naive expected-cost comparison
($C_s + pD < C_r$) by a term the folk argument omits: $\Delta\Pi$, the
\emph{pivot salvage} of the robust artifact on the refuted branch.
The correction bites hardest exactly where the folk argument feels
strongest: as $p \to 0$, debt's advantage does \emph{not}
grow without bound; it converges to $(C_r - C_s) - \delta\Delta\Pi$, which is negative
whenever salvage exceeds the discounted cost premium. The practical
distinction is between \emph{experiments run to kill} a hypothesis
(failure terminates; salvage $\approx 0$; debt strongly favored) and
\emph{experiments run to steer} (failure redirects; robust components
survive the pivot; robustness self-finances through its embedded
switching option). ``Move fast and break things'' is a theorem about
the first kind of experiment and a fallacy about the second.

Section~\ref{sec:sequential} is the paper's core: a finite-horizon
dynamic program over the venture's belief $q_t$ (updated by Bayes from
experiment outcomes) and accumulated debt stock $D_t$, with runway as
option maturity. Four results follow. \emph{(i)} A unit of promised
repayment costs not its face value but its \textbf{shadow price}
$\varphi(q) = \mathbb{E}[\delta^{\tau}\mathbf{1}\{\text{commit}\}] < 1$,
the risk-discounted probability that the venture ever reaches the
repayment boundary; in the calibration of
Figure~\ref{fig:dp} (all parameters and solution code committed),
$\varphi(0.3) \approx 0.23$ (finite-difference evaluation across the model's debt increment; illustrative calibration, as throughout):
at an even-odds-against belief, a promise
to refactor costs less than a quarter of its face value.
\emph{(ii)} A \textbf{technical-debt overhang}
(Proposition~\ref{prop:overhang}): the belief required to commit and
scale rises with the debt stock ($\bar q$ climbs from $0.57$ to $0.88$
across the calibrated debt range): accumulated shortcuts raise the
evidential bar the product must clear. We borrow the name from
\citeauthor{myers_1977}'s (\citeyear{myers_1977}) debt overhang while
flagging that the mechanism here is optimal stopping, not a
multi-claimant agency conflict (\S\ref{sec:sequential}).
\emph{(iii)} A \textbf{refactoring-pivot theorem}
(Proposition~\ref{prop:pivot}): with no carrying cost, early repayment
is weakly dominated; optimal repayment concentrates at the
commitment boundary. The signature practitioners report (a refactoring
burst when a product finds market fit) is thus \emph{derived} as
optimal option exercise rather than assumed; and because, to our
knowledge, no repository study has measured the pattern as an
event-time phenomenon, its timing enters the empirical program as a
registered prediction rather than a redescription of lore.
With a carrying cost (interest) the theorem yields a threshold:
repay early only when the interest rate exceeds the
survival-adjusted discount, the formal content of the practitioner
heuristic ``pay down high-interest debt first.''
\emph{(iv)} A \textbf{volatility result}
(Proposition~\ref{prop:vega}): the debt build holds a call (strike
$D$) where the robust build holds the underlying, so under
risk-neutral valuation a mean-preserving spread in the validated
product's value raises debt's relative advantage. Uncertainty is not
merely a circumstance in which debt is excusable; it is the input that
prices the debt option \emph{up}, though, as
Remark~\ref{rem:riskneutral} makes explicit, this prices the
instrument and does not predict the behavior of an undiversified,
risk-averse founder.

Section~\ref{sec:danger} develops the toxicity boundary, including two
classes the informal discourse misses: \emph{contagious} debt, whose
repayment cost grows endogenously as other components build against
the shortcut (named there as an informal extension; the sequential
model's debt increment is exogenous, and formalizing endogenous growth
is future work), and \emph{epistemic} debt, shortcuts that corrupt
the experiment's signal quality (no instrumentation, flaky demos) and
thereby destroy the very information the debt was incurred to buy.
Section~\ref{sec:empirical} converts the propositions into a phased,
pre-registerable empirical program whose design confronts the
survivorship bias that would otherwise be fatal to repository studies
in this area. Section~\ref{sec:discussion} discusses limitations
honestly: this is a theory paper; the calibration is illustrative; the
program is specified, not executed.

\subsection{Position in a research program}

This paper is the second in a two-paper program on the mispricing of
software costs. Its sibling theorizes \emph{structural waste}: the
recurring, unconditional operational friction of a fragmented system
arrangement: cost that accrues per period held, independent of any
future decision, and that nobody chose as an instrument. The
demarcation between the two constructs is temporal and requires no
external reference to state: debt is a \emph{contingent} liability of
a past decision; waste is an \emph{unconditional} present rent. The
present paper exploits that demarcation from the pricing side: it is
precisely the contingency of debt's cost that makes the option pricing
work, and precisely the loss of contingency that defines toxicity. Figure~\ref{fig:quadrants} arranges
the program: cost structure (contingent vs.\ unconditional) against
origination (chosen vs.\ emergent). Strategic debt and structural
waste occupy opposite corners; toxic debt is the chosen instrument
that has degenerated into an unconditional rent; economically, a
self-inflicted structural waste.

\begin{figure}[t]
\centering
\includegraphics[width=0.86\textwidth]{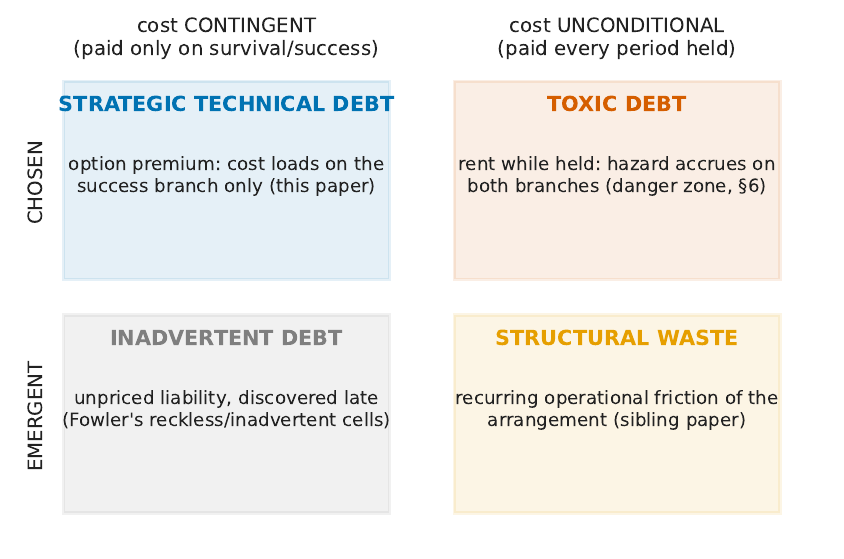}
\caption{The program map. Columns: whether the item's cost is
contingent on survival/success or accrues unconditionally while held.
Rows: whether the item was deliberately chosen or emerged unmanaged.
This paper prices the upper-left cell and locates the danger zone in
the upper-right; the lower-right (unmanaged, unconditional friction) is the structural-waste cell, theorized in the program's sibling
paper.}
\label{fig:quadrants}
\end{figure}

% =====================================================================
\section{Related work}\label{sec:related}

\subsection{Technical debt: from metaphor to measured liability}

\citet{cunningham_1992} introduced debt as a deliberate,
communicable trade-off. The maturation of the field (\citealp{kruchten_2012}'s theory-and-practice agenda, the mapping study
of \citet{li_2015}, the Dagstuhl consolidation
\citealp{avgeriou_2016}) built taxonomies of debt types and a
management lifecycle: identify, measure, prioritize, repay
\citep{guo_seaman_2011,yli_huumo_2016}. The financial vocabulary is
pervasive (principal, interest, portfolios;
\citealp{ampatzoglou_2015}) but, as that review itself documents, it
is used \emph{descriptively}; the instruments are metaphors, not
models. Closest to us formally is \citet{schmid_2013}, whose
expected-cost treatment includes the probability that a debt item's
evolution cost is ever incurred. Our two-period rule
(Proposition~\ref{prop:rule}) contains Schmid's insight as the special
case $\delta = 1, \Delta\Pi = 0$; the sequential model then adds what an
expected-cost calculation cannot express: an exercise \emph{policy}
(when to repay, when to abandon), state-dependent pricing
($\varphi(q)$), and derived timing predictions. Practitioner
treatments that endorse deliberate debt
\citep{mcconnell_2007,fowler_2009,allman_2012,buschmann_2011} supply
the intuition and the taxonomy axes but no payoff structure; our
cost-loading criterion (Definition~\ref{def:loading}) is the formal
version of their prudent/reckless line, relocated from the
decision-maker's mind to the instrument's payoff.

\subsection{Real options in software design}

\citet{sullivan_1999} first framed software design decisions as
staged investments under uncertainty valued as real options;
\citet{erdogmus_favaro_2002} carried the machinery into process
economics, valuing agile flexibility and the deferral principle as
options; \citet{baldwin_clark_2000} showed modularity creates a
portfolio of options on module-level substitution. This line prices
\emph{architecture}, the flexibility good design buys. We price the
opposite instrument: the liability \emph{bad} design incurs, and show
it too has option structure, on the borrower's side. The corporate-
finance machinery we import is standard: contingent-claims valuation
\citep{black_scholes_1973,cox_1979}, investment under uncertainty and
the option value of waiting \citep{dixit_pindyck_1994,
mcdonald_siegel_1986,trigeorgis_1996}, and debt overhang
\citep{myers_1977}. Within software engineering, the options lens has
also been pointed at refactoring itself: \citet{bahsoon_emmerich_2004}
value, in Black--Scholes style, the architectural flexibility a
refactoring \emph{buys}, still pricing flexibility gained, not the
liability incurred. Two clarifications discipline what follows. First,
our machinery is the dynamic-programming tradition of
\citet{dixit_pindyck_1994} (expectations under true probabilities
with an exogenous discount), not market-spanning contingent-claims
pricing, which hypothesis-validation risk (idiosyncratic and
non-traded) does not license. Second, although we call
Proposition~\ref{prop:overhang} an overhang, its mechanism is an
exercise-threshold comparative static, not
\citeauthor{myers_1977}'s multi-claimant agency conflict; the
distinction is drawn explicitly in Section~\ref{sec:sequential}.

\subsection{Entrepreneurship as experimentation}

That early ventures are best understood as sequences of information-
buying experiments is by now the central organizing idea of the
empirical entrepreneurship literature: \citet{kerr_2014} state the
frame; \citet{camuffo_2020} show experimentally that founders trained
to treat decisions as hypothesis tests perform better;
\citet{koning_2022} show A/B-testing infrastructure shifts startup
outcomes; \citet{gans_2019} build strategy on the choice among
experiments; \citet{ries_2011} is the practitioner canon; and the
declining cost of running experiments has itself been formalized as a
driver of venture financing
\citep{nanda_rhodeskropf_2013,ewens_nanda_rhodeskropf_2018}. On the
software side, \citet{giardino_2014,giardino_2016} document the
``greenfield'' pattern (startups deliberately privileging speed over
engineering quality while searching for fit) and \citet{klotins_2018}
confirm deliberate debt accumulation in start-up practice at scale.

Nor is the pairing of real-options reasoning with technical debt
itself new: \citet{abad_ruhe_2015} apply real-options thinking as a
decision-support method for requirements-level debt, and
\citet{alzaghoul_bahsoon_2013,alzaghoul_bahsoon_2014} value the option
to substitute cloud services under a debt constraint. What none of
these, nor \citet{schmid_2013}'s expected-cost model, nor the
architecture-flexibility line running from \citet{sullivan_1999} and
\citet{baldwin_clark_2000} through \citet{bahsoon_emmerich_2004}, provides is a \emph{priced contingent claim}: a treatment in which the
repayment obligation is itself an option with a derived exercise
policy, a state-dependent shadow price below face value, an overhang
effect, and a timing prediction for \emph{when} repayment is optimal
rather than assumed. The clearest near-neighbor on the timing question
is a contemporaneous preprint, \citet{colla_2026}, which argues from
backlog opportunity cost (not a contingent-repayment option) that
immediate remediation is not generally optimal; this paper reaches
that conclusion independently, through a different mechanism, with a
derived rather than assumed decision boundary. The gap this paper
fills is therefore precise: no existing treatment prices technical
debt's repayment as a \emph{contingent claim on venture survival},
with the shadow-price, overhang, timing, and volatility consequences
that follow from that pricing.

% =====================================================================
\section{The construct: strategic technical debt}\label{sec:construct}

\subsection{Definitions}

\begin{definition}[Experiment]\label{def:experiment}
An \emph{experiment} is a built software artifact deployed to obtain a
signal $S$ about a product hypothesis $H$, characterized by its cost
of construction, its signal quality $(\alpha, \beta)$, where $\alpha = \Pr(S{=}\mathrm{pass} \mid H)$,
$\beta = \Pr(S{=}\mathrm{pass} \mid \neg H)$, and the terminal
disposition of the artifact (scaled, salvaged, or discarded).
\end{definition}

\begin{definition}[Strategic technical debt]\label{def:std}
\emph{Strategic technical debt} is an implementation liability
$(C_s, D)$ incurred deliberately in constructing an experiment, where
$C_s$ is the reduced construction cost and $D$ the repayment
(refactoring, re-implementation, hardening) required to convert the
artifact to a scalable asset, such that repayment is
\emph{contingent}: $D$ is paid only on the continuation path.
\end{definition}

\begin{definition}[Cost loading]\label{def:loading}
A debt item's cost is \emph{success-loaded} when its expected cost
conditional on the venture's death is zero (the obligation is
extinguished by abandonment) and \emph{unconditional} when it imposes
expected cost per period held, independent of the venture's eventual
fate. A debt item is \textbf{strategic} only if success-loaded;
an item with a material unconditional component is \textbf{toxic}.
\end{definition}

The operational test for Definition~\ref{def:loading} is a single
question, askable per item at incurrence: \emph{``If the venture died
tomorrow, would any of this item's cost still have been paid?''} A
hardcoded configuration, a monolith that should be services, a
single-tenant assumption: if the venture dies, no one ever pays;
success-loaded; strategic. An unpatched injection vulnerability, an
unbacked-up production database, silent handling of regulated data. These can detonate \emph{while the experiment is running}, on both
branches; their expected cost accrues per period held; toxic. The
demarcation is a property of the payoff structure, not of the
engineer's diligence, which is what makes it falsifiable
(Section~\ref{sec:empirical}, EP5) where prudent/reckless is not.

\begin{definition}[Epistemic debt]\label{def:epistemic}
\emph{Epistemic debt} is a construction shortcut that degrades the
experiment's signal quality (lowering $\alpha$ or raising $\beta$) rather than (or in addition to) creating a repayment liability:
absent instrumentation, unrepresentative hacks in the measured path,
demo instability that contaminates user response.
\end{definition}

Epistemic debt occupies a special position: its cost is neither $C$
nor $D$ but a degradation of the option's \emph{underlying}: the
information the experiment exists to buy. In the model of
Section~\ref{sec:sequential} the value of experimentation is
increasing in signal informativeness, so epistemic debt can carry
negative value even when free. This aligns the construct with the
experimental-discipline results of \citet{camuffo_2020} and
\citet{koning_2022}: the one part of the build a strategic borrower
must \emph{not} shortcut is the measurement apparatus.

\subsection{Demarcation}

\emph{Against Fowler's quadrant} \citep{fowler_2009}: our first axis
(chosen vs.\ emergent) coincides with deliberate/inadvertent; our
second does not coincide with prudent/reckless but formalizes it: prudence becomes success-loading of the cost, a checkable payoff
property.

\emph{Against structural waste} \citep{azarang_2026_sw}: the sibling
construct is the recurring friction of a system arrangement, paid
every period regardless of any rework decision: definitionally
unconditional, definitionally unchosen. The demarcation the sibling
paper draws temporally (liability vs.\ rent) reappears here as the
strategic/toxic boundary \emph{within} debt: toxic debt is debt whose
economics have degenerated to rent. The four-cell map is
Figure~\ref{fig:quadrants}.

% =====================================================================
\section{The two-period model}\label{sec:twoperiod}

\subsection{Setup}

A venture holds hypothesis $H$ with prior $p = \Pr(H)$. It will run
one experiment, observe the outcome, and then continue or abandon.
Values are in units of the venture's cost numeraire; $\delta \in
(0,1]$ discounts the second period.

\begin{assumption}\label{ass:costs}
The experiment can be built debt-laden at cost $C_s$ or robust at
cost $C_r$, with $0 < C_s < C_r$. The debt-laden artifact carries
repayment obligation $D > 0$; the robust artifact carries none.
\end{assumption}

\begin{assumption}\label{ass:signal}
The experiment is decisive: it reveals $H$. (Imperfect signals are
deferred to Section~\ref{sec:sequential}, where they drive the belief
dynamics.)
\end{assumption}

\begin{assumption}\label{ass:payoffs}
If $H$ is validated, continuation is worth $V$ gross (net of scale-up
costs common to both builds), with $V > D$: on the validated branch,
converting the debt-laden artifact is worth doing. If $H$ is refuted,
continuation is worthless, but each artifact has salvage value: components, infrastructure, and subsystems redeployable in the
venture's pivot: $\Pi_r \in [0, C_r)$ for the robust build and
$\Pi_s \in [0, \Pi_r]$ for the debt-laden build. Write
$\Delta\Pi = \Pi_r - \Pi_s \ge 0$ for the salvage differential.
An earlier draft set $\Pi_s = 0$; nothing below needs that, and
debt-laden code does in practice salvage something on a pivot. Every
result depends on salvage only through $\Delta\Pi$.
\end{assumption}

\begin{assumption}\label{ass:risk}
The venture is risk-neutral. (Founder risk preference and financing
frictions are discussed in Section~\ref{sec:discussion}.)
\end{assumption}

Figure~\ref{fig:tree} displays the tree. The value of each build
strategy:
\begin{align}
W_s &= -C_s + \delta\left[\, p\,(V - D) + (1-p)\,\Pi_s \,\right], \label{eq:ws}\\
W_r &= -C_r + \delta\left[\, p\,V + (1-p)\,\Pi_r \,\right]. \label{eq:wr}
\end{align}

\begin{figure}[t]
\centering
\includegraphics[width=0.95\textwidth]{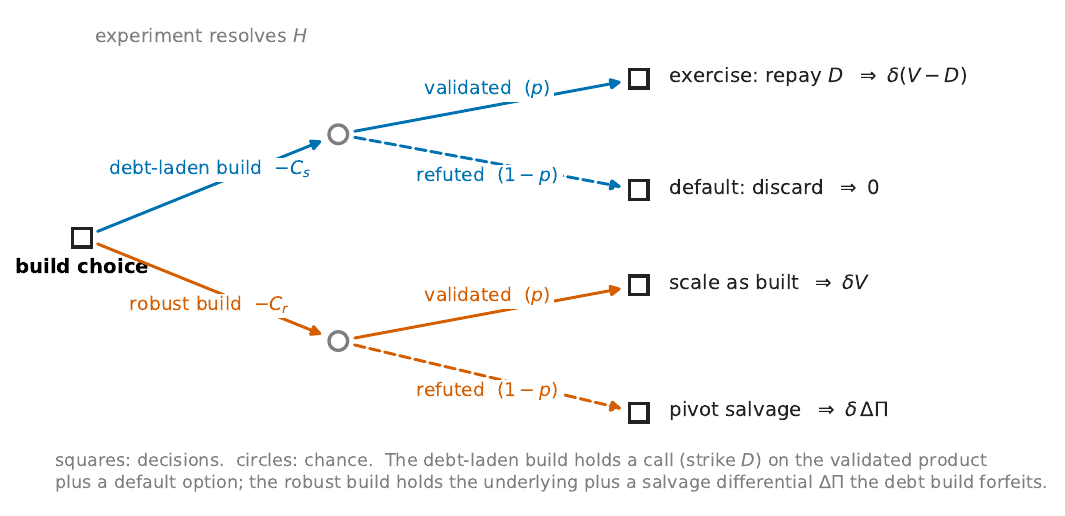}
\caption{The two-period decision tree. The debt-laden build purchases
a call on the validated product (strike $D$) plus a costless default
option; the robust build purchases the underlying plus the salvage
differential $\Delta\Pi$ on the refuted branch (each build salvages
its own $\Pi_i$; only the differential prices the choice). Generated by \texttt{figs.py}.}
\label{fig:tree}
\end{figure}

\subsection{The exercise-adjusted debt rule}

\begin{proposition}[Debt rule]\label{prop:rule}
Under Assumptions~\ref{ass:costs}--\ref{ass:risk}, the debt-laden
build is optimal iff
\begin{equation}\label{eq:rule}
C_r - C_s \;>\; \delta\left[\, p\,D \;+\; (1-p)\,\Delta\Pi \,\right].
\end{equation}
\end{proposition}

\begin{proof}
Immediate from $W_s - W_r =
(C_r - C_s) - \delta pD - \delta(1-p)\Delta\Pi$.
\end{proof}

The left side is the premium robustness charges today; the right side
is what debt costs in expectation: the repayment $D$, priced not at
face value but at $\delta p$ (the two-period shadow price of debt) \emph{plus} the salvage differential $\Delta\Pi$ the debt build
forfeits, priced at $\delta(1-p)$. Three corollaries organize the paper's disagreements
with the folk argument.

\begin{corollary}[The naive rule as a special case]\label{cor:naive}
With $\delta = 1$ and $\Delta\Pi = 0$, \eqref{eq:rule} reduces to
$C_s + pD < C_r$, the expected-cost comparison of the informal
argument and, in essence, of \citet{schmid_2013}. The naive rule
systematically overstates the case for debt whenever the venture's
failure mode is a pivot rather than a death.
\end{corollary}

\begin{corollary}[Kill vs.\ steer]\label{cor:kill}
As $p \to 0$, $W_s - W_r \to (C_r - C_s) - \delta\Delta\Pi$. Debt is
dominated at \emph{all} repayment levels for all
$p < \bar p \equiv 1 - \frac{C_r - C_s}{\delta\,\Delta\Pi}$ whenever
$\Delta\Pi > (C_r - C_s)/\delta$. Maximum uncertainty argues for maximum
debt only in \emph{experiments run to kill}, where refutation
terminates the line of work and salvage is nil. In \emph{experiments
run to steer} (where refutation redirects the venture and robust
components survive the pivot at a higher rate than debt-laden ones), robustness self-finances through its salvage \emph{differential}, and
the folk rule inverts. The inversion requires only $\Delta\Pi$ large,
not $\Pi_s = 0$: granting the debt build its own salvage weakens
nothing. Feasibility under the payoff bounds ($\Delta\Pi \le \Pi_r <
C_r$) requires $C_s > C_r(1-\delta)$, mild at realistic discounting but a genuine domain restriction: the inversion cannot occur when the
robust build costs nearly nothing to replace.
\end{corollary}

\begin{corollary}[Indifference belief]\label{cor:pbar}
For $D > \Delta\Pi$, debt is optimal iff
$p < p^\ast = \frac{(C_r - C_s)/\delta \;-\; \Delta\Pi}{D - \Delta\Pi}$:
the strategy is a low-belief instrument, abandoned in favor of
robustness as evidence accumulates: a first, static glimpse of the
repayment dynamics formalized in Proposition~\ref{prop:pivot}.
\end{corollary}

Figure~\ref{fig:regions} plots the boundary \eqref{eq:rule} in the
$(p, D)$ plane for salvage-free and salvage-rich experiments; the
qualitative lesson is that the region where debt is optimal is carved
away from the \emph{low}-$p$ side by $\Delta\Pi$, not the high-$p$ side.

\begin{figure}[t]
\centering
\includegraphics[width=0.9\textwidth]{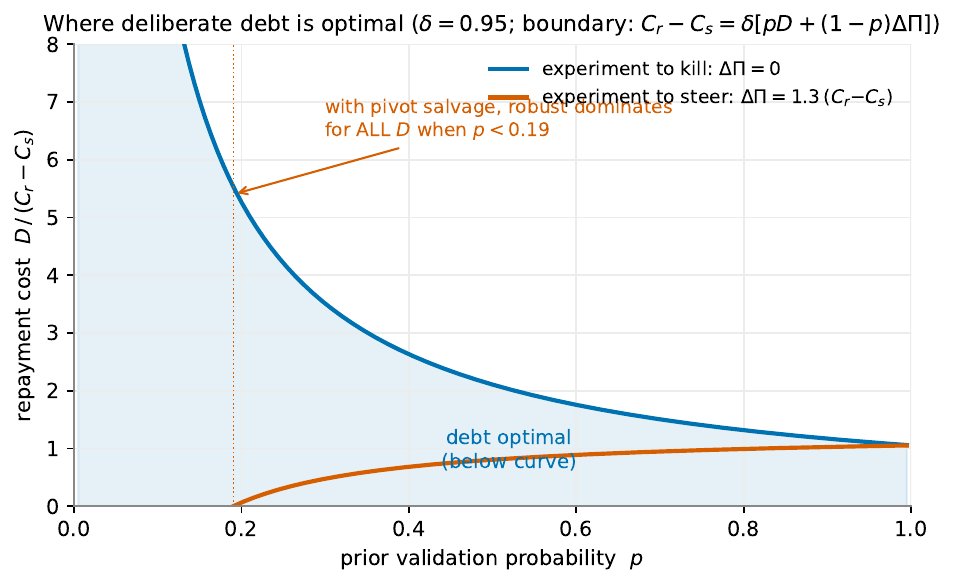}
\caption{The region where the debt-laden build is optimal
(Proposition~\ref{prop:rule}), for an experiment run to kill
($\Delta\Pi = 0$) and an experiment run to steer
($\Delta\Pi = 1.3\,(C_r - C_s)$). Pivot salvage removes the low-$p$ corner
entirely: below $p \approx 0.19$ robustness dominates at every
repayment level. Admissibility: $\Delta\Pi \le \Pi_r < C_r$ caps the
steer curve at cost ratios $C_s/C_r > 1 - 1/1.3 \approx 0.23$; the
diagram therefore illustrates a higher-salvage cost regime than the
dynamic calibration of \S5 ($c_s/c_r = 0.15$), under which the
admissible differential is capped at $\Delta\Pi < 1.18\,(C_r - C_s)$.
Generated by \texttt{figs.py}.}
\label{fig:regions}
\end{figure}

\begin{remark}[Repayment is capped by rebuilding]\label{rem:cap}
$D$ never exceeds the cost of rebuilding robustly with validated
knowledge: effective repayment is $D_{\mathrm{eff}} = \min(D,
C_r^{\mathrm{rebuild}})$, and $C_r^{\mathrm{rebuild}} \le C_r$ because
the rebuild inherits a validated specification. The much-feared
``rewrite from scratch'' is, in this frame, the exercise of a
\emph{ceiling} on the option's strike, an argument that the
worst case of strategic debt is better than its reputation, provided
the debt is genuinely disposable. The cap is, however, a
\emph{pre-launch} result: once the artifact carries live traffic, a
rebuild inherits data migration and undocumented load-bearing
behavior, and $C_r^{\mathrm{rebuild}}$ can exceed $C_r$; and
Proposition~\ref{prop:pivot} locates repayment exactly at the
commitment boundary, where traffic begins. We flag the tension rather
than hide it.
\end{remark}

\subsection{Volatility prices the debt option up}

\begin{proposition}[Volatility]\label{prop:vega}
Let the validated continuation value be a random variable $\tilde V$
with $\mathbb{E}[\tilde V] = V$, realized after validation but before
the repayment decision, and extend \eqref{eq:ws}--\eqref{eq:wr} with
$W_s = -C_s + \delta p\, \mathbb{E}[\max(\tilde V - D, 0)]$ and
$W_r = -C_r + \delta p\,\mathbb{E}[\tilde V] + \delta(1-p)\Pi_r$, and
$W_s$ correspondingly carries $\Pi_s$. Then,
under risk-neutral valuation (Assumption~\ref{ass:risk}), a
mean-preserving spread of $\tilde V$ weakly increases $W_s - W_r$,
strictly if the spread transfers mass across the kink at $D$.
\end{proposition}

\begin{proof}
$\mathbb{E}[\tilde V]$ is invariant under mean-preserving spreads;
$\mathbb{E}[\max(\tilde V - D, 0)]$ is the expectation of a convex
function and weakly increases, strictly when the spread moves mass
across the kink at $D$ \citep[cf.][]{cox_1979}. (Mass added below $D$
that does not cross the kink leaves the call value unchanged; the
kink-crossing condition is both necessary and sufficient for
strictness.)
\end{proof}

\begin{remark}[Valuation, not behavior]\label{rem:riskneutral}
Proposition~\ref{prop:vega} is the result most sensitive to
Assumption~\ref{ass:risk}, and the statement scopes it deliberately:
it prices the instrument under risk neutrality. For a risk-averse
founder a mean-preserving spread raises the call's price \emph{and}
its risk, and numerically (CARA utility over the same payoffs) the
expected-utility comparison can invert at every plausible
risk-aversion level. The proposition therefore predicts the behavior
of diversified capital, not of an undiversified founder; the
cross-sectional prediction EP4 conditions on financing structure
accordingly.
\end{remark}

The debt build holds a call; the robust build holds the underlying.
Volatility in what a validated product would be worth (the defining condition of new markets) is therefore not merely a circumstance in
which debt is \emph{excusable}; it is priced \emph{into} the debt
option's favor. This is the formal content of the intuition that
``hacky'' building belongs to high-variance opportunity spaces, and
it yields the cross-sectional prediction EP4.

% =====================================================================
\section{The sequential model}\label{sec:sequential}

The two-period model prices one experiment. Real early-stage work is a
\emph{sequence}: each experiment updates belief, each build decision
adds or avoids debt, and the venture chooses continually among
experimenting cheaply, experimenting robustly, repaying, committing,
and abandoning against a finite runway. This section develops the
resulting dynamic program; it is the paper's core.

\subsection{Primitives}

Time $t = 0, 1, \dots, T$ (runway of $T$ experiments, the option's maturity). State: belief $q_t = \Pr(H \mid \text{history})$ and debt
stock $D_t$. Each period the venture chooses one action:
\begin{itemize}[itemsep=1pt]
\item \textbf{abandon}: value $0$ (all debt extinguished);
\item \textbf{commit}: pay scale-up cost $K$ and repay $D_t$, receive
      $q_t V$: terminal value $q_t V - K - D_t$;
\item \textbf{experiment cheap}: pay $c_s$, debt grows to
      $D_t + d$, observe $S \in \{\mathrm{pass}, \mathrm{fail}\}$
      with $\Pr(\mathrm{pass}) = q_t\alpha + (1-q_t)\beta$, update
      $q_{t+1}$ by Bayes;
\item \textbf{experiment robust}: pay $c_r > c_s$, debt unchanged,
      same signal and update.
\end{itemize}
The Bellman equation, with $\tau = T - t$ experiments remaining:
\begin{equation}\label{eq:bellman}
J_\tau(q, D) = \max\!\Big\{\,0,\;\; qV - K - D,\;\;
 -c_s + \delta\,\mathbb{E}\big[J_{\tau-1}(q', D + d)\big],\;\;
 -c_r + \delta\,\mathbb{E}\big[J_{\tau-1}(q', D)\big] \Big\},
\end{equation}
with $J_0(q, D) = \max(0,\, qV - K - D)$. A separate \emph{repay}
action (pay $x \le D$ mid-course) is analyzed in
Proposition~\ref{prop:pivot}; carrying costs (interest and toxicity)
enter in Section~\ref{sec:danger}.

\begin{lemma}[Shape]\label{lem:shape}
$J_\tau(q, D)$ is convex and nondecreasing in $q$, nonincreasing in
$D$, and satisfies the \emph{deferral bound}
\begin{equation}\label{eq:lipschitz}
J_\tau(q, D) - J_\tau(q, D + x) \;\le\; \varphi_\tau(q, D)\, x
\;\le\; x
\quad\text{for } x > 0,
\end{equation}
where $\varphi_\tau(q, D) =
\mathbb{E}\big[\delta^{\tau_c}\,\mathbf{1}\{\text{commit}\}\big]$
is the expected discounted indicator of ever committing under the
optimal policy from $(q, D)$. (Proof sketch:
Appendix~\ref{app:proofs}.)
\end{lemma}

\begin{definition}[Shadow price of debt]\label{def:phi}
$\varphi_\tau(q, D) \in [0, 1]$ is the \emph{shadow price of
technical debt}: the present cost of one additional unit of promised
repayment. In the two-period model $\varphi = \delta p$; in general it
is the risk-discounted probability that the venture survives to the
repayment boundary.
\end{definition}

$\varphi$ is the workhorse quantity of the analysis. As a mathematical
object it is standard, and we name the lineage rather than claim it:
the expected discounted indicator of a stopping event is a \emph{state price}: the present value of one unit paid at the random commitment time, the same reduced-form machinery that prices defaultable
claims. (For a contemporaneous shadow-pricing derivation in an exploration-and-stopping program, see \citealp{sannikov_zhong_2026}.) The contribution is not the object but what its engineering
comparative statics say (Propositions~\ref{prop:overhang}
and~\ref{prop:pivot}). It is small exactly when
the venture is uncertain: Figure~\ref{fig:dp}B shows
$\varphi(q{=}0.3) \approx 0.23$ (evaluated, as throughout, as the finite difference $[J(q,0)-J(q,d)]/d$ across the debt increment, the natural shadow price on a discrete debt grid) and $\varphi \to 1$ only inside the commitment region; and it is the correct discount at which an
engineering team should book a promise to refactor. The folk
discourse prices debt at face value (hence the moral vocabulary);
the naive optimists price it at zero (hence the horror stories);
the model prices it at $\varphi$.

\begin{figure}[t]
\centering
\includegraphics[width=\textwidth]{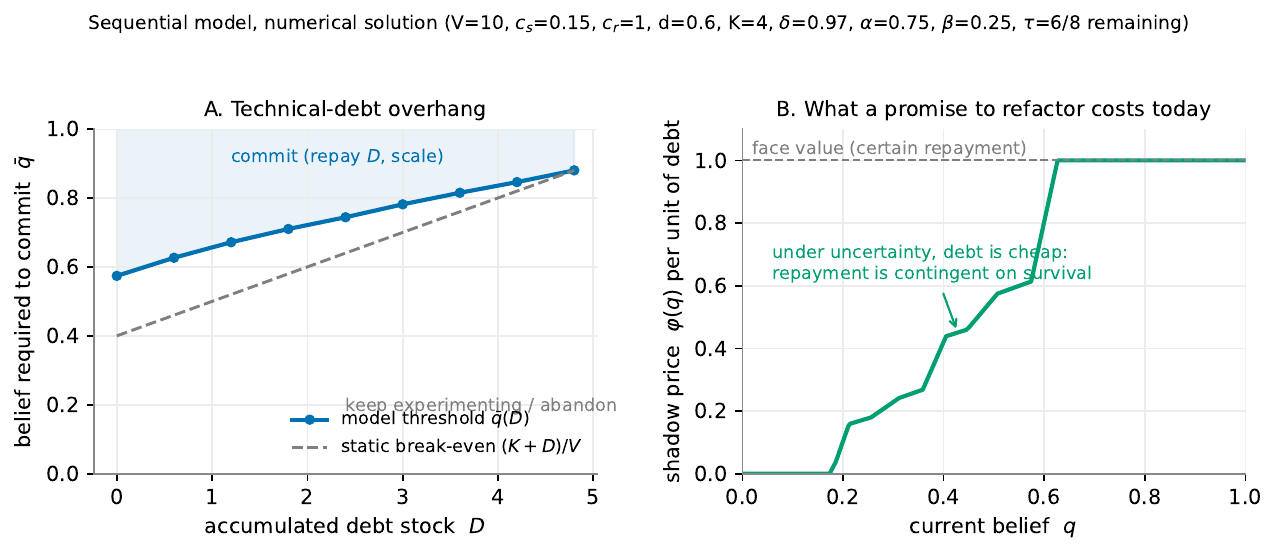}
\caption{Numerical solution of \eqref{eq:bellman} by backward
induction on a belief grid ($n_q = 2001$; parameters in the panel
title; solver committed in \texttt{figs.py}). \textbf{A.}\ The
commitment threshold $\bar q(D)$ rises with the debt stock: the technical-debt overhang (Proposition~\ref{prop:overhang}); the static
break-even $(K + D)/V$ understates the bar because committing also
forfeits the option to keep experimenting. \textbf{B.}\ The shadow
price $\varphi(q)$ of one unit of debt: far below face value under
uncertainty, reaching face value only in the commitment region. The
calibration is illustrative, not estimated.}
\label{fig:dp}
\end{figure}

\subsection{Technical-debt overhang}

\begin{proposition}[Overhang]\label{prop:overhang}
The commitment threshold
$\bar q_\tau(D) = \inf\{q : \text{commit is optimal at } (q, D)\}$
is nondecreasing in $D$, and strictly increasing wherever the
continuation option has positive value. (Proof:
Appendix~\ref{app:proofs}, closed by an \emph{envelope Lipschitz lemma}: the value function's upper slope in $q$ is bounded by $V$,
proved via the affine-invariance of the Bayes operator, which makes
the commit set an upper interval at every $(\tau, D)$. Numerically
corroborated: the threshold is monotone in every configuration of an
independent 91-configuration sensitivity sweep, and in
Figure~\ref{fig:dp}A $\bar q$ rises from $0.574$ at $D = 0$ to
$0.880$ at $D = 4.8$.)
\end{proposition}

We call this an overhang, and the temptation is to cite
\citet{myers_1977} and stop; the honest comparison is narrower.
Myers's overhang is a multi-claimant \emph{agency} result: an external
creditor's fixed claim dilutes equityholders' incentive to invest.
This model has no second claimant. The introduction said it plainly:
the creditor is the venture's own future. What
Proposition~\ref{prop:overhang} exhibits is an exercise-threshold
comparative static in the
\citet{mcdonald_siegel_1986}/\citet{dixit_pindyck_1994} tradition: a
higher strike raises the optimal-exercise trigger, amplified by the
forfeited option to keep experimenting. The name is borrowed for its intuition: promised repayment deters the very investment (scaling)
that would give the promise value; but the mechanism is optimal
stopping, not wealth transfer. The software reading is a caution the pro-debt
argument must carry on its own books: debt is cheap to \emph{incur}
under uncertainty precisely because repayment is improbable, but each
unit incurred \emph{raises the evidential bar} the product must later
clear, because committing now means buying the repayment at face
value. A venture that has borrowed heavily needs a stronger signal to
justify scaling; and, per Figure~\ref{fig:dp}A, the gap between the
model threshold and the static break-even is the price of the
forfeited option to keep experimenting.

\subsection{The refactoring-pivot theorem}

\begin{proposition}[Repayment timing]\label{prop:pivot}
Extend \eqref{eq:bellman} with a repay action: pay $x \in (0, D]$ now,
continue in state $(q, D - x)$.
\textbf{(a)} With no carrying cost, repaying before commitment is
weakly dominated: paying $x$ now costs $x$; deferring costs
$\varphi\,x \le x$, strictly less whenever $\varphi < 1$. Optimal
repayment concentrates at the commitment boundary.
\textbf{(b)} With a per-period carrying cost $\rho$ on the
outstanding stock (interest: slowed iteration, onboarding drag),
immediate repayment of a unit becomes optimal iff the discounted
expected interest stream saved exceeds the deferral discount:
$\rho \cdot \mathbb{E}\big[\sum_{u \ge 0} \delta^u
\mathbf{1}\{\text{debt still held at } u\}\big] \;>\; 1 - \varphi$.
(Proof: Appendix~\ref{app:proofs}.)
\end{proposition}

Part (a) turns an anecdote into a theorem with a timestamp. That
teams refactor heavily when a product finds market fit is widely
reported by practitioners but has, to our knowledge, never been
measured as an event-time phenomenon in repository data; here it is
derived as the \emph{optimal exercise policy of the conversion
option}: repayment is the strike, and one does not pay a strike on an
option one may yet allow to expire. The empirically loaded content is
the \emph{absence} half: away from validation events, mid-course
repayment of disposable debt should be rare in ventures behaving
consistently with the model, a prediction (EP3$'$) that repository
data can refute, with the timing of validation events observable as
funding rounds and the repayment burst observable with refactoring
detection \citep{tsantalis_2018,silva_2016}. Part (b) is the formal
version of ``pay down high-interest debt first'': interest justifies
early repayment only when it outruns the survival discount; and since carrying costs scale with how much the team must \emph{touch}
the indebted code, it simultaneously rationalizes the practitioner
rule of repaying debt on hot paths and ignoring it in cold corners.

\subsection{Signal quality and epistemic debt}

The value of every continuation branch of \eqref{eq:bellman} is
nondecreasing in informativeness (higher $\alpha$, lower $\beta$
refine the posterior fan; a Blackwell-dominated experiment is worth
weakly less). Epistemic debt (Definition~\ref{def:epistemic})
degrades $(\alpha, \beta)$ and therefore attacks $J$ directly, independent of any repayment. A shortcut that saves $c$ but blurs the
signal can carry negative net value even at $\varphi = 0$; no
repayment holiday redeems an experiment that cannot discriminate. In
the taxonomy of Section~\ref{sec:danger} this makes instrumentation
the one component of an experimental build that is never a legitimate
site for debt, which is, we note, exactly the discipline
empirically associated with founder outperformance
\citep{camuffo_2020,koning_2022}.

% =====================================================================
\section{The danger zone: toxicity as loss of contingency}
\label{sec:danger}

Let a candidate debt item impose an unconditional expected loss
$\ell > 0$ per period held: a hazard rate $\lambda$ of an event
costing $L$ ($\ell = \lambda L$): breach of an unpatched
vulnerability, loss of unbacked-up data, a compliance penalty. The
defining feature is that $\ell$ accrues on \emph{both} branches,
survival and death alike, from incurrence until repair.

\begin{proposition}[Toxicity]\label{prop:toxic}
\textbf{(a)} If $\ell \ge c_r - c_s$, the cheap-with-toxic build is
strictly dominated by the robust build at every state $(q, D)$, every
horizon, and every prior: the per-period saving is exhausted by the
unconditional loss before any option value accrues.
\textbf{(b)} For every hazard $\lambda > 0$ and loss $L > 0$, the
toxic build is weakly dominated by its \emph{clean twin}, the otherwise-identical build (same cost $c_s$, same signal, same debt
increment) without the hazard, at every state, every horizon, and
all $(p, V, D, T)$; strictly wherever continuation value is positive.
(Proof: Appendix~\ref{app:proofs}.)
\end{proposition}

\begin{remark}[The right baseline]\label{rem:baseline}
Part (b) deliberately compares the toxic build against its clean twin,
not against the robust build. The stronger-sounding universal, ``no
parameters make the toxic build optimal \emph{against the robust build},'' is false: for small $\lambda$ and short horizons the cheap
build's cost saving can outweigh a finite expected ruin loss
(numerically, at $\lambda = 0.15$, $T = 6$, the toxic-cheap build
beats the robust build on a nontrivial share of states: 7--18\%
depending on the hazard's functional form, absorbing-only versus
absorbing plus flow loss, across two independent replications; see the
verification packets cited under Data and artifact availability, which
surfaced this and supplied the repair adopted here). Toxicity is not
``never worth it relative to building robustly''; it is ``never worth
it relative to not carrying the hazard.'' The shortcut and the hazard
are separable instruments, and a rational venture takes the shortcut
without the hazard, which is exactly the necessity claim EP5 tests.
\end{remark}

Toxicity is thus not extreme strategic debt; it is a \emph{category}
failure. Strategic debt's entire economics rest on the contingency of
its cost ($\varphi < 1$); a toxic item's cost is unconditional
($\varphi$-pricing does not apply), which places it economically in
the right-hand column of Figure~\ref{fig:quadrants}: a chosen structural waste. The proposition's form (dominated for \emph{all} $p$) is a necessary-condition claim: absence of toxic debt is
necessary for the strategy's rationality, at any uncertainty level.
That is precisely the logical form Necessary Condition Analysis is
built to test \citep{dul_2016}, and EP5 registers it.

Two toxic classes deserve naming because the informal discourse
misses both:

\begin{description}[itemsep=2pt]
\item[Contagious debt.]\footnote{A recent systematic mapping names the propagation phenomenon in autonomous-agent settings ``agentic technical debt'' \citep{tukur_2026_agtd}; it is a taxonomy, not an economics, and the informal status of contagion here is unchanged by it.} A shortcut embedded in an interface that
other components build against. Its repayment cost is not a constant
$D$ but an endogenously growing $D_t$: each dependent component adds
to the eventual conversion. Contagion converts a success-loaded item
into one whose \emph{strike grows with the venture's own progress} (the option analogue of negative amortization) and bounded contagion
is the real content of the practitioner rule ``hide your hacks behind
a clean interface.'' Architectural debt's documented dominance among
debt costs \citep{besker_2018,martini_2015} is consistent with this
reading, though establishing contagion as its mechanism is an
empirical task, not a fact we may assert. One honesty note: in the
model of Section~\ref{sec:sequential} the debt increment $d$ is
exogenous, so contagion is an \emph{informal} extension here: endogenizing $D_t$ in the number of dependents, and deriving the
strategic$\to$toxic degeneration this subsection describes, is open
(\S\ref{sec:discussion}).
\item[Epistemic debt.] Definition~\ref{def:epistemic}: corruption of
the signal is a loss on the information side, unredeemable by any
repayment schedule.
\end{description}

Table~\ref{tab:taxonomy} applies the cost-loading test to common
items. Classification is, moreover, not a one-shot act. The
cost-loading answer is given at incurrence, but real artifacts drift:
prototype code becomes load-bearing without anyone re-asking the question, the most common way a strategic item silently migrates
cells. The discipline therefore carries a \emph{re-classification
trigger}: when an item acquires its first external dependent, or
survives a pre-declared number of deploys, the cost-loading question
is asked again, and an item that fails on re-ask is treated as having
migrated. The empirical program's instruments record these migrations
rather than assuming the initial label persists.

\begin{table}[t]
\centering
\caption{The cost-loading test applied to common early-stage debt
items. The classifying question: \emph{if the venture died tomorrow,
would any of this item's cost still have been paid?} The test applies
to the concrete artifact, not the category name: a category can
span cells (hardcoded feature flags are strategic; hardcoded
credentials are toxic).}
\label{tab:taxonomy}
\small
\begin{tabular}{@{}>{\raggedright\arraybackslash}p{0.36\textwidth}>{\raggedright\arraybackslash}p{0.33\textwidth}>{\raggedright\arraybackslash}p{0.23\textwidth}@{}}
\toprule
\textbf{Item} & \textbf{Cost structure} & \textbf{Class} \\
\midrule
Hardcoded configuration / single-tenant assumptions &
Repaid only on scale-up (per artifact: hardcoded \emph{credentials}
are toxic) & Strategic \\
Monolith where services will be needed &
Strike grows as components accrete against it & Contagious (strategic
only while dependents are few) \\
Skipped CI / build automation for a prototype &
Repaid only on team growth & Strategic \\
Ad-hoc manual deploys to live production &
Hazard $\ell$ scales with blast radius & Toxic-leaning \\
Skipped tests on \emph{throwaway} experiment code &
Repaid only if code survives & Strategic (borderline) \\
Skipped tests on the \emph{measurement} path &
Degrades $\alpha,\beta$ now & Epistemic (toxic) \\
No analytics / instrumentation &
Experiment cannot discriminate & Epistemic (toxic) \\
Shortcut inside a depended-upon interface &
Strike $D_t$ grows with adoption & Contagious (toxic if unbounded) \\
Unpatched security vulnerability &
Hazard $\ell$ accrues while held & Toxic \\
No backups of irreplaceable data &
Ruin hazard, both branches & Toxic \\
Silent non-compliance (regulated data) &
Hazard + unbounded $L$ & Toxic \\
\bottomrule
\end{tabular}
\end{table}

Beyond item-level toxicity, three \emph{regime} conditions bound the
strategy itself: long validation cycles push $\varphi \to 1$ (the
survival discount vanishes when repayment is near-certain before the signal arrives: enterprise and infrastructure products); pivot-rich
domains raise $\Delta\Pi$ (Corollary~\ref{cor:kill}); and platform products
maximize contagion because external developers multiply dependents.
The model thus predicts its own boundary: strategic debt is an
instrument of \emph{consumer-speed, kill-style, edge-of-system}
experimentation, and ventures outside that regime should borrow
sparingly however uncertain they are.

% =====================================================================
\section{Falsifiable propositions and the empirical program}
\label{sec:empirical}

The theory yields five empirical propositions. Two are advanced to
\emph{primary} status (slated for pre-registration and execution, not merely specification), and three are recorded as successors.
The distinction is deliberate: a program that specifies five studies
it will not run converts falsifiability into decoration, and the
primary pair was chosen because each is executable from repository
history alone, without the recall, reachability, and
unrealized-value-measurement problems that gate the successors. The
full protocol (instruments, sampling frames, sample sizes, analysis plans, and the pre-registration discipline) is specified
in the companion document \texttt{EMPIRICAL\_PROGRAM.md}; this
section states the propositions and the design commitments that make
them testable rather than narratable.

\begin{description}[itemsep=3pt]
\item[EP1 (Debt rule): successor, not promised.] Ventures whose early build choices conform to \eqref{eq:rule} (cheap builds where salvage is low and validation
odds are long, robust builds where pivot value is high) exhibit
better information-per-cost trajectories than nonconforming peers,
with the dependent variable operationalized as hypotheses brought to
disposition (validated or killed) per unit of build spend, and
controls for founder experience, team size, sector, and stage.
Instrument: retrospective founder/CTO cohort study with
item-level debt recall anchored to repository evidence.
\item[EP2 (Overhang): successor.] Among ventures reaching a validation event,
time from validation to scale commitment increases with the
accumulated debt stock at validation, controlling for team size and
funding. Instrument: repository + funding-record panel.
\item[EP3$'$ (Refactoring pivot; sharpened from the original EP3 specification): primary.] Repayment activity
(refactoring commits, test-suite construction, modularization)
concentrates in a window following \emph{validation} events, and is
correspondingly \emph{rare} during the experimentation regime;
detection by \citet{tsantalis_2018}. Event-time refactoring analysis
itself is not new, and the claim is scoped against that literature:
refactoring aligns with \emph{release} dates in some projects
\citep{guana_tsantalis_2013}, not in others \citep{hoque_2014}, and
release-wise refactoring patterns have since been characterized at
scale \citep{noei_2025}. What is unmeasured is timing around
\emph{validation} events specifically (the model's $q_t$ crossing $\bar q$), which no release calendar proxies. The design confronts
its central confound directly rather than by concordance checking:
funding rounds cause refactoring through headcount growth and
productionization mandates, mechanisms unrelated to option exercise,
so funding-as-proxy cannot separate the two causal paths. The
discriminating cells are therefore \emph{validation without funding}
(launch metrics, usage milestones, signed pilots preceding any round)
against \emph{funding without validation} (rounds raised on narrative
before usage evidence); the model predicts the refactoring burst
follows the first and not the second. The \emph{absence} clause (rarity during experimentation) is what makes this a risky
prediction rather than a redescription of lore.
\item[EP$\Pi$ (Pivot salvage): primary.] The salvage parameters
this paper prices have never been measured, and they are measurable:
for ventures whose history contains a codeable pivot, the salvage
realization is the fraction of the pre-pivot codebase surviving into
the post-pivot product, computed from repository history by file
lineage (renames and content-similarity tracked, per the perils
discipline of \citealp{kalliamvakou_2016}). Two predictions.
\emph{First}, salvage is nondegenerate: the surviving fraction is
materially above zero for a substantial share of pivots, the empirical ground for modeling $\Pi_s > 0$ rather than assuming debt
code worthless. \emph{Second}, components classifiable ex ante as
robust-built survive pivots at a higher rate than debt-classified
components within the same venture; the within-venture contrast is
the measured $\Delta\Pi > 0$ that
Corollary~\ref{cor:kill} turns on. This gives the model's most
distinctive parameter its own instrument instead of leaving it a
calibration.
\item[EP4 (Volatility): successor, not promised.] Across sectors, deliberate early debt
intensity increases with outcome-value dispersion, conditional on
cycle length and financing structure
(Remark~\ref{rem:riskneutral}). Dispersion is computed over the
\emph{full cohort-entry sample}, including ventures that die unrecorded; realized-exit variance alone would reintroduce
precisely the survivorship truncation the sampling design exists to
avoid.
\item[EP5 (Toxicity ceiling): successor.] In the $(x, y)$ plane of toxic-debt
presence vs.\ venture outcome, the upper-right corner is empty:
survival to scale is bounded above by a ceiling in toxic exposure.
Logical form: necessity, not correlation; method: Necessary Condition
Analysis \citep{dul_2016}, with the ceiling predicted from toxicity
severity and exposure-duration proxies that are \emph{not} part of
the Phase-0 classification question, so the test carries
information independent of the coder's label rather than confirming
coder consistency. Strategic-debt intensity, by contrast,
carries \emph{no} such ceiling; the contrast between the two
NCA plots is the construct validity test of
Definition~\ref{def:loading}.
\end{description}

\paragraph{The survivorship confound, confronted.} The fatal design
error available in this area is sampling on outcome: successful
ventures' repositories are visible \emph{because} they succeeded, and
``survivors refactor after funding'' cannot then distinguish optimal
exercise from mere survival. The program therefore commits to
\emph{cohort-entry sampling}: ventures enter the frame at first
public commit / accelerator intake / launch listing (before outcomes exist), and both surviving and dying ventures are followed
to disposition, with the perils of repository mining handled per
\citet{kalliamvakou_2016}. Two refinements keep the commitment
honest. First, the entry frames differ in residual selection: accelerator intake and launch listings condition on clearing a prior hurdle correlated with outcomes, so the least-selected frame
(repository-signature entry) is reserved for confirmatory tests.
Second, cohort entry does not remove every bias: private repositories
under-sample stealth ventures, and outcome-correlated deletion
right-censors failures; the program records visibility and deletion
as first-class attrition outcomes per wave rather than assuming them
away. EP3$'$'s prediction is conditional on
reaching a validation event, and the model makes the complementary
prediction for non-survivors: debt incurred, never repaid: the modal, silent, \emph{correct} outcome of the strategy, observable
only under entry sampling.

\paragraph{Discipline.} Carried over from the sibling program
\citep{azarang_2026_sw} and from this research program's standing
practice: predictions frozen and analyzers committed before data are
read; verdicts adjudicated mechanically against pre-registered
criteria; negative and ambiguous results reported as first-class.
The propositions above are stated in advance of any data collection;
no pilot data existed when they were fixed.

% =====================================================================
\section{Discussion}\label{sec:discussion}

\subsection{What the reframing buys}

The moral framing of technical debt survives because it is usually
locally correct: in mature systems, where survival is near-certain,
$\varphi \approx 1$, debt trades at face value, and the engineering
conscience prices it right. The framing fails at the boundary where
modern software increasingly lives: pre-validation, where
$\varphi$ is small and falling in uncertainty, and where the
expensive error is inverted: teams that build robustly for products
that die have paid $C_r$ for information $C_s$ would have bought.
The model's slogan-sized content: \emph{debt you repay only if you
win is an option premium; debt that can sink you while you wait is a rent}; and the whole management problem of early-stage engineering
is keeping the first kind from quietly becoming the second, through
contagion, through signal corruption, or through a regime shift the
team failed to notice.

For practice, the model reduces to four disciplines: classify every
shortcut by the cost-loading question at incurrence (Table~\ref
{tab:taxonomy}) and re-ask it on the re-classification triggers of
Section~\ref{sec:danger}; defer repayment to the validation boundary
unless carrying cost demonstrably exceeds $1 - \varphi$
(Proposition~\ref{prop:pivot}); never borrow against the
measurement apparatus (Definition~\ref{def:epistemic}); and treat the
kill-vs.-steer call (Corollary~\ref{cor:kill}) as itself uncertain: salvage is often discovered post hoc, so a team unsure of its
experiment's kind should price $\Delta\Pi$ as a random variable with a
conservative floor, not as zero.

\subsection{Debt when code is cheap}
\label{sec:agentic}

The model's parameters are not constants of nature, and the current
shift to LLM-assisted and agent-driven development moves three of them
at once, in directions the comparative statics price cleanly. First,
the robustness premium $C_r - C_s$ \emph{falls}: when code is cheap to
produce, the cheap build's savings shrink, because the robust build is
also cheap. By the debt rule~\eqref{eq:rule}, a smaller left side
justifies strictly less debt at every belief. Second, repayment cost
falls with it: refactoring is increasingly agent labor, which lowers
the effective strike and, through
Proposition~\ref{prop:pivot}\,(b), makes early repayment optimal at
lower carrying costs. Third, the hazard side of the toxicity boundary
moves the other way. The emerging empirical record on AI-generated
code finds debt-like issues accumulating faster and persisting: a
multivocal review reports LLM assistance amplifying code, design, and
documentation debt \citep{ehsani_2026}, a production-scale quality
profile of AI-generated C++ \citep{tran_2026_aicpp}, and a large-scale study of
verified AI-authored commits finds a substantial share of
AI-introduced issues surviving to repositories' latest revisions
\citep{liu_2026_aidebt}. In the model's terms, agent-speed
development lowers the price of \emph{incurring} debt while raising
the flow rate of \emph{unclassified} debt; and unclassified debt
is where toxicity hides, since nobody applied
Definition~\ref{def:loading}'s question to it. The net prediction is
sharp and testable: as build costs fall, the binding constraint on
rational debt strategy shifts from the cost side (is the shortcut
worth it?) to the classification side (is the shortcut's cost
success-loaded?). Tooling that makes code cheap without making
cost-loading legible should therefore produce exactly the pattern the
early empirical record shows: more debt, incurred faster, with a
rising unconditional component. None of this is adjudicated here; it
is the model's reading of a literature that so far has measurements
and no decision theory.

\subsection{Limitations}

This is a theory paper. The calibration in Figure~\ref{fig:dp} is illustrative: committed and reproducible but not estimated from
data; the empirical program is specified, not executed, and the
propositions could die on contact with it (EP3$'$'s absence clause and
EP5's ceiling are genuinely risky). The model is risk-neutral;
founder risk aversion and financing frictions would lower debt's
appeal, while venture portfolio logic
\citep{kerr_2014} would raise it; the net direction is an open
question. $D$ is exogenous where contagion argues it should be a
controlled state variable. Competition is absent: racing rivals
erode the option value of waiting \citep{mcdonald_siegel_1986} and
can rationalize debt beyond \eqref{eq:rule}. Human costs of debt-laden
codebases (morale, hiring, the difficulty of retaining engineers in a shantytown) enter, if at all, through $\rho$, and deserve better;
debt-driven engineer attrition in particular is structurally closer
to the ruin process of Proposition~\ref{prop:toxic} than to a linear
carrying cost, and routing it through $\rho$ understates it. Three
further limits deserve naming. The model holds per-experiment costs
constant in the debt stock, but accumulated debt taxes the velocity
of everything built on top of it; with $c_s(D_t), c_r(D_t)$ weakly
increasing, debt consumes the very runway the option needs to mature,
and the overhang and timing results should be re-derived under that
coupling. Competition deserves decomposition rather than the
concession above: racing rivals erode the option value of waiting,
which would \emph{compress} the overhang gap and can force repayment
earlier than Proposition~\ref{prop:pivot}'s boundary; competitive intensity is a moderator EP3$'$ should record, not noise. And the
empirical program's limits are measurement limits: every proxy in
Section~\ref{sec:empirical} (validation events, debt signals,
toxicity classes) carries error the pre-registration must bound, not
merely acknowledge.
Finally, the two-paper program's boundary claim (toxic debt behaves
economically like structural waste) is stated, not yet measured;
its test belongs to the sibling program's instruments.

% =====================================================================
\section{Conclusion}\label{sec:conclusion}

Technical debt began as a borrowing metaphor and spent thirty years
being moralized into a pathology. Priced properly (as a call on the
validated product with a costless default, its repayment discounted
by the probability of ever owing it), deliberate debt under
uncertainty is what Cunningham said it was: a rational instrument.
The price of taking that claim seriously is precision about where it
stops: repayment obligations that grow with your own success,
shortcuts that blind your experiments, and hazards that accrue
whether or not you survive are not strategic debt at higher leverage;
they are different instruments with rent-like economics, and no level
of uncertainty redeems them. The model derives, rather than assumes,
the observable signatures of rational borrowing: repayment bursts
at validation, an evidential bar that rises with the debt stock, debt
intensity that tracks outcome volatility; and the empirical program
registers them as predictions that can fail. Whether founders in the
wild borrow like option traders or like optimists is now a
measurement question, which is where we intend to take it.

% =====================================================================
\section*{Data and artifact availability}

This draft contains no empirical data. The numerical solutions in
Figures~\ref{fig:regions} and~\ref{fig:dp} are fully reproducible
from the committed solver (\texttt{figs.py}; deterministic, no seed
required). The empirical program, including instruments and
pre-registration templates, is versioned alongside this manuscript.
An independent verification packet is archived in the repository
(\texttt{review/2026-07-23-operon/}): a four-lens adversarial referee
review; an independent re-implementation of the solver reproducing
every headline number to machine precision, with a 91-configuration
sensitivity sweep (overhang monotonicity and repayment-deferral held
in all configurations); a three-modality prior-art sweep
(arXiv, OpenAlex, web) behind the Section~\ref{sec:related}
positioning; and per-entry bibliography verification against
Crossref/OpenAlex. The v0.1$\to$v0.2 revision, including the
correction of Proposition~\ref{prop:toxic}(b) recorded in
Remark~\ref{rem:baseline}, responds to that packet. The v0.3
envelope-Lipschitz proof was itself adversarially verified after
authoring (two independent formal referees over four attack surfaces (FOSD boundaries, the $D$-shift under the action maximum,
convexity propagation, and the $\delta = 1$ edge); verdict: sound,
no counterexample constructible; report archived in the same
directory), and its three invariants were reproduced independently
of the committed solver.

\bibliographystyle{plainnat}
\bibliography{references}

% =====================================================================
\appendix
\section{Proof sketches}\label{app:proofs}

\paragraph{Lemma~\ref{lem:shape}.}
Convexity and monotonicity in $q$: by backward induction. $J_0$ is a
maximum of affine functions of $q$. The Bayes operator maps posteriors
affinely in the prior along each signal branch, and
$\mathbb{E}[J_{\tau-1}(q', \cdot)]$, the expectation of a convex
function of the posterior over the signal distribution induced by $q$, is convex in $q$ (standard for value functions of Bayesian
stopping problems; cf.\ the optimal-stopping treatments in
\citealp{dixit_pindyck_1994}); maxima of convex functions are convex.
Monotonicity in $D$ is immediate. For the deferral bound
\eqref{eq:lipschitz}: from state $(q, D + x)$, replicate the optimal
policy of state $(q, D)$; every payoff is identical except on commit,
where the replicating policy pays $x$ more, discounted to
$\delta^{\tau_c}$ at the (random) commitment time and paid only on
the commitment event. Hence
$J(q, D + x) \ge J(q, D) - \varphi\, x$ with
$\varphi = \mathbb{E}[\delta^{\tau_c}\mathbf 1\{\text{commit}\}] \le 1$.
\hfill$\square$

\paragraph{Proposition~\ref{prop:overhang}.}
Commit value $f(q, D) = qV - K - D$ falls in $D$ at rate $1$;
continuation value falls at rate at most $\varphi \le \delta < 1$
(Lemma~\ref{lem:shape}, noting one experiment must elapse before any
commitment). Both are continuous in $q$; $f$ crosses the
continuation-or-abandon envelope from below (slope $V$ against a
value-function slope in $q$ bounded by $V$, with strict inequality
wherever continuation retains option value). An increase in $D$
therefore lowers $f$ strictly more than the envelope at every $q$,
moving the crossing right; where continuation value is zero
(abandon region boundary) the threshold is $(K + D)/V$, increasing in
$D$ directly. Monotonicity of $\bar q_\tau(D)$ follows; strictness
holds wherever the forfeited option has positive value. The
single-crossing step is closed by the following lemma.

\emph{Lemma (envelope Lipschitz bound).} For every $\tau$, $D$, and
$q' \ge q$: $J_\tau(q', D) - J_\tau(q, D) \le V(q' - q)$, and the
same bound holds for the continuation envelope $E_\tau$.

\emph{Proof.} Write
$T[h](q) = \Pr(\mathrm{pass} \mid q)\, h(q^{+}(q)) +
\Pr(\mathrm{fail} \mid q)\, h(q^{-}(q))$
for the one-experiment Bayes operator. Three facts. (i) $T$ preserves
affine functions exactly: for $h(x) = a + bx$,
$T[h](q) = a + b\,\mathbb{E}[q' \mid q] = a + bq$ by the martingale
property of Bayesian posteriors. (ii) For an informative signal
($\alpha > \beta$), the posterior distribution is FOSD-nondecreasing
in the prior: $q^{+}$ and $q^{-}$ are both nondecreasing in $q$,
with derivatives $\alpha\beta / \Pr(\mathrm{pass} \mid q)^2$ and
$(1-\alpha)(1-\beta) / \Pr(\mathrm{fail} \mid q)^2$, and the weight
on $q^{+}$ is nondecreasing; hence $T[g]$ is nonincreasing
whenever $g$ is. (At the boundaries $q \in \{0, 1\}$, and for
$\alpha, \beta \in \{0, 1\}$, monotonicity is weak; weak is all the
argument uses.) (iii) If $h$ is convex with upper slope at most $V$, then
$g := h - u$ with $u(x) = h(1) - V(1 - x)$ is nonnegative (the
Lipschitz bound applied at the right endpoint), convex, and has
$g(1) = 0$, hence is nonincreasing: a convex function with positive
slope anywhere before $1$ cannot return to $0$ at $1$. Now induct on
$\tau$. $J_0 = \max(0,\, qV - K - D)$ is convex with upper slope
$\le V$. If $J_{\tau-1}(\cdot, D')$ is convex with upper slope
$\le V$ for every $D'$, then by (i)--(iii),
$T[J_{\tau-1}] = u + T[g]$ is an affine function of slope $V$ plus a
nonincreasing function, so its upper slope is $\le V$; the experiment
actions $-c + \delta\, T[J_{\tau-1}]$ have upper slope
$\le \delta V \le V$; abandon has slope $0$; commit has slope exactly
$V$; and a pointwise maximum of functions with upper slope $\le V$
has upper slope $\le V$. Convexity propagates explicitly rather than
by citation: writing $\varphi(x, y) = (x + y)\, h\!\left(x / (x +
y)\right)$ for the perspective of the convex $h$ (jointly convex in $(x, y)$), the operator decomposes as
$T[h](q) = \varphi\big(q\alpha,\, (1-q)\beta\big) +
\varphi\big(q(1-\alpha),\, (1-q)(1-\beta)\big)$, a sum of convex
functions composed with affine maps of $q$, hence convex.
$\blacksquare$

\emph{Single crossing.} The commit payoff $f(q) = qV - K - D$ has
slope exactly $V$; every non-commit action has upper slope
$\le \delta V$. Hence $f - E_\tau$ is nondecreasing (strictly
increasing for $\delta < 1$), and the commit set
$\{q : f \ge E_\tau\}$ is an upper interval $[\bar q, 1]$, so
$\bar q_\tau(D) = \inf\{q : \text{commit}\}$ reads as ``commit for
all $q \ge \bar q$.'' At $\delta = 1$ the qualification is explicit:
$f - E_\tau$ can plateau where an experiment action ties the commit
slope at $V$, so the crossing may be a flat interval rather than a
point; $\bar q$ is then the infimum of a weakly-crossed set, the
upper-interval conclusion and the $D$-monotonicity survive
unchanged, and only strictness and uniqueness of the crossing
require $\delta < 1$. With the $D$-monotonicity established above,
Proposition~\ref{prop:overhang} is proved. The lemma's content is
additionally checked numerically on the calibrated solution: the
maximum upper slope of $J$ over all $(\tau, D)$ slices equals $V$ to
machine precision, the commit set is an upper interval in all $121$
slices, and $T$'s affine-invariance holds to $9 \times 10^{-16}$.
\hfill$\square$

\paragraph{Proposition~\ref{prop:pivot}.}
(a) Compare policies from any non-commit state: (i) repay $x$ now,
then follow the optimal policy of $(q, D - x)$; (ii) hold, follow the
same action path, and pay the extra $x$ if and when that path
commits. Path (ii)'s cost for the marginal $x$ is
$\varphi x \le x$ = path (i)'s cost, with strict inequality whenever $\varphi < 1$, that is, whenever $\delta < 1$ or commitment is
neither certain nor immediate. (At $\delta = 1$ with certain
commitment the comparison is an exact indifference; the numerical
sweep exhibits the machine-precision tie.) Hence early repayment is
weakly dominated and strictly so off the commitment boundary for
$\delta < 1$.
(b) Add carrying cost $\rho x$ per period while $x$ is outstanding.
Early repayment saves the stream
$\rho x \,\mathbb{E}[\sum_{u\ge0}\delta^u \mathbf 1\{\text{held at }u\}]$
and forgoes the deferral discount $(1 - \varphi)x$; the stated
threshold is the comparison of the two. \hfill$\square$

\paragraph{Proposition~\ref{prop:toxic}.}
(a) Couple the two builds path-by-path: cheap-with-toxic pays
$c_s + \ell \ge c_r$ per period in expectation while generating an
identical signal process and a weakly larger repayment obligation;
its continuation set is therefore weakly smaller and its per-period
flow weakly worse, strictly in the flow. Dominance holds state-wise,
hence in value. (b) Couple the toxic build with its clean twin
path-by-path: identical costs, identical signal process, identical
debt trajectory; the only difference is the hazard, which on every
path weakly reduces value: each period the item is held multiplies
continuation by $(1 - \lambda)$ and adds expected loss $-\lambda L \le 0$, so
$J^{\mathrm{toxic}} = J^{\mathrm{clean}} - (\text{a nonnegative
term})$, with strict inequality wherever continuation value is
positive. Dominance is state-wise and universal in
$(p, V, D, T, \lambda)$. (A prior draft compared the toxic build
against the \emph{robust} build and claimed universal dominance; that
comparison fails at small $\lambda$ and short horizons; see Remark~\ref{rem:baseline}.) \hfill$\square$

\end{document}